 \documentclass[10pt, singlecolumn, twoside, romanappendices]{IEEEtran}

\IEEEoverridecommandlockouts

\usepackage{graphicx,cite,acronym,amsmath,amssymb,amsfonts,color,subfigure,floatflt,stfloats,bm,textgreek,multirow,amsthm}
\usepackage{epsfig,epstopdf,psfrag}

\usepackage[table]{xcolor}
\usepackage{bm}

\usepackage[cmintegrals]{newtxmath}
\usepackage{graphics,hhline,array,cite,bbm}

\usepackage{latexsym}
\usepackage{times}
\usepackage{makeidx}

\usepackage{mathtools}
\usepackage{graphicx}
\usepackage{hyperref}
\usepackage{xcolor}
\usepackage{enumitem}

\usepackage{tikz}                    
\usetikzlibrary{arrows,positioning}  

\DeclareGraphicsExtensions{.png,.jpg,.eps}

\acrodef{ESPAR}{electrically steerable passive array radiator}

\acrodef{RA}{reconfigurable antenna}

\acrodef{SIM}{stacked intelligent metasurface}

\acrodef{SE}{spectral efficiency}

\acrodef{DAC}{digital-to-analog converter}

\acrodef{DMA}{dynamic metasurface antenna}

\acrodef{SINR}{signal-to-interference noise ratio}

\acrodef{ESIT}{electromagnetic signal and information theory} 

\acrodef{ELAA}{extremely large antenna arrays} 

\acrodef{DSA}{dynamic scattering array}

\acrodef{ULA}{uniform linear array}

\acrodef{UCA}{uniform circolar array}

\acrodef{IIoT}{industrial Internet-of-things}

\acrodef{IT}{information theory}

\acrodef{SRE}{smart radio environment}

\acrodef{EMO}{electromagnetic object}

\acrodef{SVD}{singular value decomposition}

\acrodef{PSWF}{prolate spheroidal wave function}

\acrodef{CR}{channel response}

\acrodef{BS}{base station}

\acrodef{MS}{mobile station}

\acrodef{UE}{user equipment}

\acrodef{MIMO}{multiple-input multiple-output}

\acrodef{MISO}{multiple-input single-output}

\acrodef{RIS}{reconfigurable intelligent surface}

\acrodef{IRS}{intelligent reconfigurable surface}

\acrodef{LIS}{large intelligent surface}

\acrodef{MIS}{medium intelligent surface}

\acrodef{SIS}{small intelligent surface}

\acrodef{DoF}{degrees-of-freedom}

\acrodef{AF}{amplify \& forward}

\acrodef{DF}{detect \& forward}

\acrodef{JF}{just forward}

\acrodef{CSI}{channel state information}

\acrodef{RV}{random variable}

\acrodef{i.i.d.}{independent, identically distributed}

\acrodef{PSD}{power spectral density}

\acrodef{PDF}{probability distribution function}

\acrodef{CDF}{cumulative distribution function}

\acrodef{ch.f.}{characteristic function}

\acrodef{AWGN}{additive white Gaussian noise}

\acrodef{RSSI}{received signal strength indicator}

\acrodef{SNR}{signal-to-noise ratio}

\acrodef{LRT}{likelihood ratio test}

\acrodef{GLRT}{generalized likelihood ratio test}

\acrodef{GML}{generalized maximum likelihood}

\acrodef{LOS}{line-of-sight}

\acrodef{NLOS}{non-line-of-sight}

\acrodef{GDOP}{geometric dilution of precision}

\acrodef{GPS}{Global Positioning System}

\acrodef{FIM}{Fisher information matrix}

\acrodef{PEB}{position error bound}

\acrodef{WSN}{Wireless Sensor Network}

\acrodef{MAC}{medium access control}

\acrodef{RSS}{received signal strength}

\acrodef{RTT}{round-trip time}

\acrodef{MIMO}{multiple-input multiple-output}

\acrodef{MF}{matched filter}

\acrodef{ED}{energy detector}

\acrodef{ML}{maximum likelihood}

\acrodef{NL}{nonlinear}

\acrodef{MSE}{mean square error}

\acrodef{RMSE}{root mean square error}

\acrodef{ppm}{part-per-million}

\acrodef{PRP}{pulse repetition period}

\acrodef{ACK}{acknowledge}

\acrodef{UWB}{ultrawide bandwidth}

\acrodef{TNR}{threshold-to-noise ratio}

\acrodef{LOS}{line-of-sight}

\acrodef{LS}{least squares}

\acrodef{IR-UWB}{impulse radio UWB}

\acrodef{FCC}{Federal Communications Commission}

\acrodef{TH}{time-hopping}

\acrodef{PPM}{pulse position modulation}

\acrodef{PAM}{pulse amplitude modulation}

\acrodef{MUI}{multi-user interference}

\acrodef{PDP}{power delay profile}

\acrodef{PPP}{Poisson point process}

\acrodef{DS}{delay spread}

\acrodef{CED}{channel excess delay}

\acrodef{BPZF}{band-pass zonal filter}

\acrodef{SIR}{signal-to-interference ratio}

\acrodef{RFID}{radio frequency identification}

\acrodef{WPAN}{wireless personal area networks}

\acrodef{WWLB}{Weiss-Weinstein lower bound}

\acrodef{DP}{direct path}

\acrodef{MF}{matched filter}

\acrodef{MMSE}{minimum-mean-square-error}

\acrodef{SBS}{serial backward search}

\acrodef{NBI}{narrowband interference}

\acrodef{WBI}{wideband interference}

\acrodef{INR}{interference-to-noise ratio}

\acrodef{CIR}{channel impulse response}

\acrodef{ISI}{inter-symbol interference}

\acrodef{CPR}{channel pulse response}

\acrodef{LRT}{likelihood ratio test}

\acrodef{MUI}{multi-user interference}

\acrodef{EM}{electromagnetic}

\acrodef{CW}{continuous wave}

\acrodef{RF}{radiofrequency}

\acrodef{FCC}{Federal Communications Commission}

\acrodef{EIRP}{effective radiated isotropic power}

\acrodef{RCS}{radar cross section}

\acrodef{BAV}{balanced antipodal Vivaldi}

\acrodef{PRake}{partial Rake}

\acrodef{RTLS}{real time locating system}

\acrodef{CRB}{Cram\'{e}r-Rao bound}

\acrodef{ZZB}{Ziv-Zakai bound}

\acrodef{TOA}{time-of-arrival}

\acrodef{TOF}{time-of-flight}

\acrodef{WSN}{wireless sensor network}

\acrodef{MAC}{medium access control}

\acrodef{RSS}{received signal strength}

\acrodef{TDOA}{time difference-of-arrival}

\acrodef{RF}{radiofrequency}

\acrodef{RTT}{round-trip time}

\acrodef{AOA}{angle-of-arrival}

\acrodef{MF}{matched filter}

\acrodef{ED}{energy detector}

\acrodef{ML}{maximum likelihood}

\acrodef{MUR}{Multistatic radar}

\acrodef{HDSA}{high-definition situation-aware}

\acrodef{RRC}{root raised cosine}

\acrodef{OFDM}{orthogonal frequency division multiplexing}

\acrodef{IF}{intermediate frequency}

\acrodef{PHY}{physical layer}

\acrodef{S-V}{Saleh-Valenzuela}

\acrodef{UHF}{ultra-high frequency}

\acrodef{PR}{pseudo-random}

\acrodef{SoC}{System on Chip}

\acrodef{SoP}{System on Package}

\acrodef{SPMF}{Single-Path Matched Filter}

\acrodef{IMF}{Ideal Matched Filter}

\acrodef{SCR}{signal-to-clutter ratio}

\acrodef{BEP}{bit error probability}

\acrodef{BER}{bit error rate}

\acrodef{WSR}{wireless sensor radar}

\acrodef{HPBW}{half power beam width}

\acrodef{LEO}{localization error outage}

\acrodef{WSS}{wide-sense stationary}

\acrodef{TR}{time-reversal}

\acrodef{WSSUS}{WSS with uncorrelated scattering}

\acrodef{GP}{Gaussian process}

\acrodef{IMU}{inertial measurement unit}

\newtheorem{theorem}{Theorem}[section]
\newtheorem{lemma}[theorem]{Lemma}

\newtheorem{corollary}[theorem]{Corollary}
\newtheorem{definition}[theorem]{Definition}

\newcommand{\rank}[1]{{\rm rank} \left \{ #1 \right \}}
\newcommand{\diag}[1]{{\rm diag} \left \{ #1 \right \}}
\newcommand{\image}[1]{{\rm image} \left \{ #1 \right \}}

\newcommand{\Real}[1]{\Re \left \{ #1\right \}}
\newcommand{\Imag}[1]{\Im \left \{ #1\right \}}

\newcommand{\EX}[1] {{\mathbb{E}}\left\{{#1}\right\}}

\newcommand{\ctranspose}{^{\mathsf{H}}}
\newcommand{\transpose}{^{\mathsf{T}}}

\newcommand{\boldA} {{\bf{A}}}

\newcommand{\boldp} {{\bf{p}}}
\newcommand{\bolde} {{\bf{e}}}

\newcommand{\boldu} {{\bf{u}}}

\newcommand{\boldX} {{\bf{X}}}

\newcommand{\boldE} {{\bf{E}}}
\newcommand{\boldJ} {{\bf{J}}}

\newcommand{\boldF} {{\bf{F}}}

\newcommand{\boldR} {{\bf{R}}}

\newcommand{\boldZ} {{\bf{Z}}}

\newcommand{\boldI} {{\bf{I}}}
\newcommand{\boldr} {{\bf{r}}}

\newcommand{\boldy} {{\bf{y}}}

\newcommand{\boldi} {{\bf{i}}}
\newcommand{\boldv} {{\bf{v}}}

\newcommand{\Ncal} {\mathcal{N}}
\newcommand{\Mcal} {\mathcal{M}}

\newcommand{\Reals} {\mathbb{R}}
\newcommand{\Complex} {\mathbb{C}}
\newcommand{\Sphere} {\mathbb{S}}

\newcommand{\Prad} {P_{\text{rad}}}

\newcommand{\Pd} {P_{\text{d}}}

\newcommand{\etad} {\eta_{\text{d}}}

\newcommand{\Rr} {R_{\text{r}}}
\newcommand{\Rd} {R_{\text{d}}}
\newcommand{\bRd} {{\bf{R}_{\text{d}}}}
\newcommand{\bRt} {{\bf{R}_{\text{t}}}}

\newcommand{\Zl} {{\bf{Z}_{\text{L}}}}

\newcommand{\Ns} {N_{\text{s}}}

\newcommand{\Parameters} {\boldsymbol{\varphi}}
\newcommand{\Parameter} {\varphi}

\newcommand{\EMsymb} {\mathsf}

\newcommand{\Em} {{\EMsymb E}}

\newcommand{\Psim} {{\EMsymb \boldsymbol{\Psi}}}

\newcommand{\boldzero} {\boldsymbol{0}}

\newcommand{\JT} {\boldJ_{T}}
\newcommand{\JtT} {\boldJ_{\mathring{T}}}

\newcommand{\JTc} {\boldJ_{T}^{(\text{c})}}

\begin{document}
\title{Fundamental Theorems on Controllability in Wave-domain Processing for Holographic MIMO}

\author{
\IEEEauthorblockN{Davide~Dardari,~\IEEEmembership{Fellow,~IEEE}}
\IEEEcompsocitemizethanks{\IEEEcompsocthanksitem 
 D.~Dardari is with the 
   Dipartimento di Ingegneria dell'Energia Elettrica e dell'Informazione ``Guglielmo Marconi"  (DEI), WiLAB-CNIT, 
   University of Bologna, Cesena Campus, 
   Cesena (FC), Italy, (e-mail: davide.dardari@unibo.it). 
    }
}

\maketitle



\begin{abstract}

Wave-domain processing is an emerging paradigm where signal processing operations are partially shifted from the digital to the electromagnetic (EM) domain. Leveraging reconfigurable EM devices, this approach aims to reduce complexity, energy consumption, and latency in next-generation wireless systems employing holographic MIMO.
This paper establishes fundamental theorems on the controllability of generic reconfigurable EM devices, where wave processing is achieved through the dynamic configuration of passive scatterers. Specifically, we derive necessary and sufficient conditions for controllability as a function of geometry and mutual coupling between elements. Finally, we provide a detailed discussion and numerical results characterizing the interplay between the number of elements, physical size, degrees of freedom, and directivity.  


\end{abstract}


\section{Introduction}

It is broadly acknowledged that meeting the demanding communication and sensing requirements of next-generation wireless networks will necessitate deploying highly dense and/or large-scale massive \ac{MIMO} antenna systems \cite{BjoChaHeaMarMezSanRusCasJunDem:24,Pre:J24,BjoEldLarLozPoo:23,YouCaiLiuDiRDumYen:25}. This evolution aims to fully harness the potential of the wireless channel, for example, by exploiting the additional \ac{DoF} available in the near-field region of antenna arrays \cite{Dar:J20,ZhaShlGuiDarEld:J23}.
Yet, massive antenna arrays also bring substantial challenges related to hardware complexity and energy consumption, raising concerns about the sustainability of future wireless infrastructures. In a fully digital transceiver design, every antenna element is paired with its own \ac{RF} chain, and all signal processing takes place in the digital domain. This setup inevitably constrains achievable data rates, energy efficiency, and latency, primarily due to the limitations of digital clock speeds and the sampling rates of ADCs and DACs.
To address these issues, hybrid analog–digital architectures have been extensively explored in \ac{MIMO} systems. By incorporating networks of phase shifters and combiners between the digital baseband, \ac{RF} chains, and antenna elements, these architectures significantly reduce the number of required \ac{RF} chains and the corresponding digital processing load. However, this reduction comes at the expense of system flexibility and reconfigurability \cite{AlkElALeuHea:14}.

Emerging technologies in reconfigurable antennas and metamaterials introduce a novel research avenue: migrating part of signal processing from the digital domain to the \ac{EM} domain (wave domain). A sustainable strategy involves distributing these tasks across digital, analog, and wave domains, forming what is termed the \emph{tri-hybrid} architecture \cite{HeaCarVikCasAkrCha:26, CasYanChaHea:26}.
Analog-domain computing, in particular, has seen renewed interest due to its energy-efficient and massively parallel capabilities. It leverages the mathematical equivalence between microwave network equations and those of common \ac{MIMO} tasks, like \ac{MMSE} estimation \cite{NerCle:25a}.
On the other hand, wave-domain processing delegates signal tasks directly to the \ac{EM} level, under the \ac{ESIT} paradigm \cite{DarTorPasDec:J26, ZhuWanDaiDebPoo:24,BjoChaHeaMarMezSanRusCasJunDem:24}.

Looking back, a key milestone in reconfigurable antennas is reactively controlled arrays, proposed by Harrington \cite{Har:78}. These use passive scatterers reconfigured via programmable reactive loads to alter radiation characteristics. This concept evolved into \acp{ESPAR} antennas, extensively studied for programmable radiation patterns featuring a single \ac{RF} chain surrounded by passive scatterers \cite{KalKanPap:13, BucJuaKamSib:20, HanZhaSheLiChiMur:21}.
More recently, advances in materials and technologies have enabled novel antenna structures tailored for next-generation wireless systems, offering higher performance, flexibility, wider frequency coverage, and lower costs \cite{GanLiuDiRJorGonShoCui:26}. This is often referred to as holographic \ac{MIMO} emerging as a key paradigm for \ac{EM}-domain processing \cite{GonGavJiHuaAleWeiZhaDebPooYue:24,DarDec:J21,GanLiuDiRJorGonShoCui:26}. 

Prominent examples include transmitarrays, fluid antennas, and \acp{RA} \cite{WonNewHaoTonCha:23, PriHarBlaKieKusFriPraMorSmi:04, RodCetJof:14, RonMarKan:22, ZhaRaoMinLiChiWonTonRos:24}. 
Another notable case is the \ac{DMA}, a traveling-wave antenna with reconfigurable subwavelength apertures. Requiring just one \ac{RF} chain per row in planar arrays, \acp{DMA} provide a flexible alternative with fewer \ac{RF} chains \cite{ShlAleImaYonSmi:21, ZhaShlGuiDarEld:J23}.
Among these technologies, \acp{SIM} have garnered particular attention \cite{ZeQiaXinWeiTie:24,AnYueGuaDiRDebPooHan:24,HasAnDiRDebYue:24,AnXuetAl5:25,SheUllSolAleCha:26,AnYueDirBenDebHan:26,FabTorDar:C25}. The \ac{SIM} concept draws from advances in optics and microwave circuits that perform machine learning via layered \ac{EM} surfaces mimicking deep neural networks for image processing \cite{ZeQiaXinWeiTie:24}. A typical \ac{SIM} comprises a sealed vacuum chamber housing stacked metasurface layers. The first layer is an active planar antenna array fed by multiple \ac{RF} chains. The \ac{EM} wave it generates propagates through subsequent layers of reconfigurable meta-atoms, with the final layer radiating the processed \ac{EM} field outward. By tuning the transmission properties of each meta-atom, the \ac{SIM} shapes the propagating \ac{EM} wave into customized waveforms. Its layered design supports backpropagation-like optimization. Optimization techniques for \acp{SIM} enabling diverse \ac{EM} wave processing tasks appear in several works such as \cite{AnYueGuaDiRDebPooHan:24,HasAnDiRDebYue:24,AnXuNgAleHuaYueHan:23,AnXuetAl5:25}. Nonlinear extensions are explored in \cite{FabTorDar:C25}.

While early studies employed simplified models for reconfigurable \ac{EM} devices, recent work uses rigorous multi-port circuit models for their balance of physical accuracy and computational tractability \cite{NerSheLiDiRCle:24,AbrBarToc:25,DesCasKhoDuVisHea:26, CasYanChaHea:26}.
Notably, most of these reconfigurable \ac{EM} technologies share a core principle: active elements interacting with tunable passive scatterers. This extends Harrington's original idea across varied geometries, element coupling, number of inputs, and implementations \cite{Har:78}.
Starting from the consideration mentioned above, a general framework for wave-domain processing, termed \ac{DSA}, accounting for multiple inputs, arbitrary geometries, and coupling effects has been recently proposed \cite{Dar:C24,Dar:J26}. A \ac{DSA} features a few active antenna elements, each linked to an \ac{RF} chain (input), surrounded by numerous reconfigurable passive scatterers. These interact in the reactive near-field, enabling joint \ac{EM} processing and ``over-the-air'' radiation. Devices like \acp{SIM}, \acp{RA}, and \acp{ESPAR} represent special cases. For example, a \ac{SIM} is a particular \ac{DSA} with a block lower-bidiagonal impedance matrix. Unlike \acp{SIM}, \acp{DSA} allow arbitrary scatterer distributions and couplings, supporting simultaneous processing and radiation for greater flexibility in compact designs \cite{Dar:J26}.
Hereafter, we use \ac{DSA} to encompass these devices (e.g., \acp{SIM}, \acp{RA}, \acp{ESPAR}) as special cases.

The literature on \acp{DSA} and related particular cases has focused on modeling and numerical optimization for specific functionalities. Theoretical performance bounds for wave-domain processing remain underexplored. In this direction, the author of \cite{Hou:26} proposes a framework to bound the fidelity of realizing a target operator via reconfigurable \ac{EM} systems. This bound is computed numerically and applied to maximize coupling in \ac{MIMO} systems using \acp{RIS} and \acp{DMA}.

To the best of the author's knowledge, no prior work addresses this question: \emph{Under what conditions is a reconfigurable \ac{EM} device controllable, that is, for any desired response, does there exist at least one configuration that realizes it?}
Verifying controllability is challenging. In fact, even though the \ac{EM} device's response is inherently linear, it depends nonlinearly on a very large number of configuration parameters.

\subsection{Our Contribution}

This paper addresses the aforementioned research question by establishing fundamental theorems on the controllability of generic reconfigurable \ac{EM} devices.

The first major contribution is the derivation of a fundamental theorem, based on differential geometry arguments, that provides a sufficient condition for the controllability of a generic reconfigurable linear \ac{EM} device, under mild assumptions regarding its transfer matrix. The strength of this theorem is that finding just one configuration that satisfies a specific expression guarantees full system controllability. Furthermore, this theorem extends well beyond the analysis of \acp{DSA}; its applicability to generic transfer functions makes it a broadly useful mathematical result.

The second contribution is the adaptation of this theorem to single-input \acp{DSA}, where wave processing is achieved through the dynamic configuration of passive scatterers. We prove that the governing equation of the \ac{DSA} satisfies the general theorem's assumptions. Consequently, we derive a closed-form sufficient condition for controllability as a function of the dimension of the observation space, the number of elements, and their geometry-dependent mutual coupling. For any specific device, whether characterized analytically, through full-wave simulations, or via measurements, this condition enables the instantaneous verification of controllability.

Finally, we provide a detailed discussion and numerical examples illustrating the interplay among the number of elements, physical size, \ac{DoF}, and directivity.

\subsection{Notation and Definitions} 
Boldface capital letters are matrices (e.g., $\boldA$), boldface lower case letters are vectors (e.g., $\boldy$), and $\boldI_N$ is the identity matrix of size $N$.
With reference to a generic matrix $\boldA$,  $a_{n,m}=[\boldA]_{n,m}$ represents its $(n,m)$th element, whereas $[\boldA]_{:,m}$ and $[\boldA]_{n,:}$ are, respectively, the $m$th column and the $n$th row of  $\boldA$. $\boldA\transpose$ and $\boldA\ctranspose$ indicate the transpose and the conjugate transpose of $\boldA$, respectively. We denote with $\left \|  \boldy  \right \|$ the  norm of vector $\boldy$. 
Define $\bolde_n = (0,0,1,0,\dots,0)^\top \in \Reals^N$ and $\boldE_n=\bolde_n \, \bolde_n\transpose=\diag{\bolde_n}$ being the $n$th standard basis diagonal matrix.
We indicate with $\Real{z}$, $\Imag{z}$, and $z^*$, respectively,  the real part,  imaginary part, and complex conjugate of the complex number $z$, and $\jmath$ the imaginary unit.
The statistical expectation of a random variable $x$ is indicated as $\EX{x}$.
Given the coordinate vector $\boldr$, $\hat{\boldr}=\boldr/r$ represents the unit vector in the radial direction and $r=\|\boldr\|$.
Denote with $\eta=377\,$Ohm the free-space impedance.

\begin{figure}[!t]
\centering\includegraphics[width=1\columnwidth]{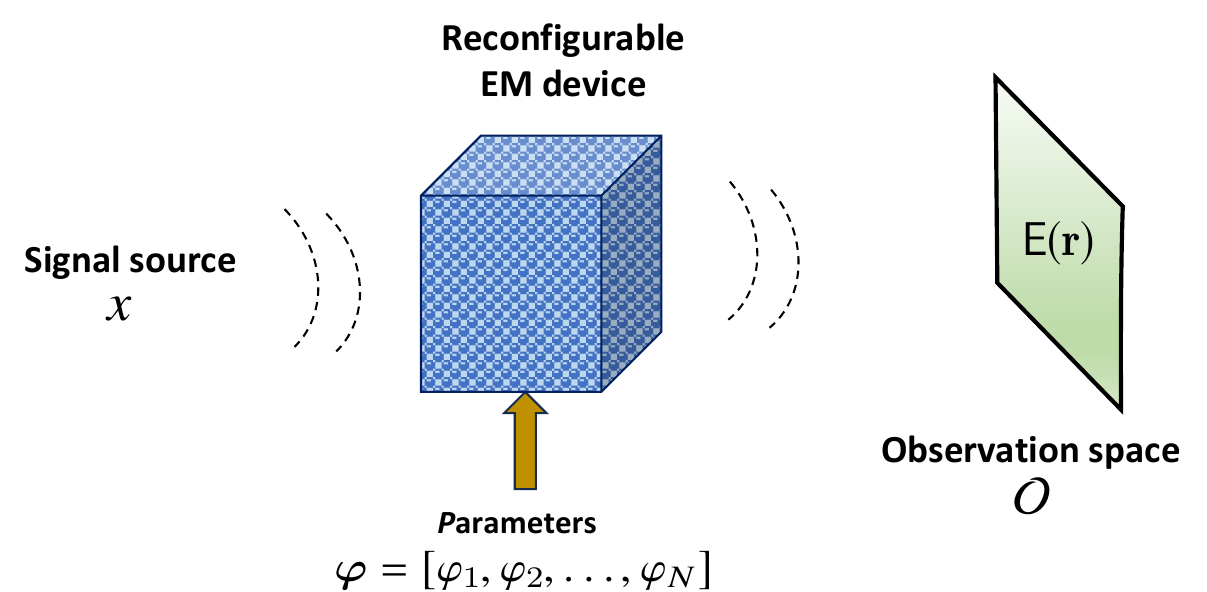}
\caption{Generic reconfigurable EM device.} 
\label{Fig:EMdevice}
\end{figure}

\subsection{Paper Organization}
The rest of the paper is organized as follows: 
Sec.~\ref{Sec:EMdevice} considers a generic reconfigurable \ac{EM} device and derives the necessary and sufficient conditions for its controllability. Starting from the $N$-port circuit model of a \ac{DSA}, the general result of Sec.~\ref{Sec:EMdevice} is particularized to obtain the controllability conditions for \acp{DSA} in Sec.~\ref{Sec:Model}. The detailed proofs of the theorems enounced in the previous sections are provided in Sec.~\ref{Sec:proof}. Starting from the characterization of radiating modes in unbounded 2D and 3D \ac{EM} spaces discussed in Sec.~\ref{Sec:Modes}, Sec.~\ref{Sec:NumericalResults} presents some numerical results to corroborate the theoretical findings. Finally, the conclusions are drawn in Sec.~\ref{Sec:Conclusion}.

\section{Controllability of Reconfigurable EM Devices}
\label{Sec:EMdevice}

In this section, we provide  sufficient conditions under which a generic linear reconfigurable \ac{EM} device is controllable. In Sec. \ref{Sec:Model}, such a result will be particularized to the controllability of \acp{DSA}.

\subsection{Response of a Generic Reconfigurable EM Device}

Consider a generic linear reconfigurable \ac{EM} device, as shown in Fig. \ref{Fig:EMdevice}, whose response to the generic source signal $x$ can be configured through the set of $N$ real parameters $\Parameters=[\Parameter_1, \Parameter_2, \ldots, \Parameters_N] \in \Reals^N$.  
In particular, given an observation space of interest $\mathcal{O}\subseteq \Reals^{3}$ in the radiating region of the \ac{EM} device, the electric field $\Em(\boldr)$ measured in $\mathcal{O}$ is given by $\Em(\boldr)=f(\boldr;\Parameters)\, x$, where $f(\boldr;\Parameters)$ denotes the  reconfigurable response of the device depending on $\Parameters$. We consider monochromatic or narrowband excitations centered at frequency $f_0$ corresponding to wavelength $\lambda$.\footnote{Through the text, a time dependence $e^{\jmath 2 \pi f_0 t}$ is implied.}  
In general, $x$ can represent an impinging \ac{EM} wave or the input port of the device.
We restrict our analysis to the case of a scalar (random) input $x$, i.e., a single source carrying the information. We leave the extension to multiple inputs for future work.

Consider now a vector basis set $\{\Psim_m(\boldr) \}_{m=1}^{\infty}$, with $\boldr \in \mathcal{O}$, for all feasible electric field distributions in the observation space $\mathcal{O}$. As it will be discussed later, the actual number of \ac{DoF} of the radiated \ac{EM} field is bounded depending on the size of the observation space and the \ac{EM} device; therefore, only a limited number $M$ of basis functions (\emph{modes}) is sufficient to represent the field with the required accuracy \cite{Dar:J24}. Accordingly, the field within the observation space  can be expressed with the series expansion
\begin{equation} \label{eq:TotE}
	\Em(\boldr) \simeq\sum_{m=1}^M  y_m \, \Psim_m(\boldr) \, .
\end{equation}
Therefore, the vector $\boldy=[y_1, y_2, \ldots, y_M]\transpose$ is representative of the field produced by the \ac{EM} device and observed in $\mathcal{O}$.

For convenience, we define the following input-output response (mapping function) 
\begin{equation} \label{eq:yT}
\boldy=T(\Parameters) \, x  
\end{equation}
which can be obtained by projecting $f(\boldr;\Parameters)$ into the $M$ basis functions $\{ \Psim_m(\boldr) \}$. 
Note that the \ac{EM} device response to the input $x$  in \eqref{eq:yT} is linear, whereas the dependence of $T(\Parameters)  \in \Complex^{M \times 1}$ on the parameters $\Parameters$ is highly nonlinear. In the following, without loss of generality, we will assume $x=1$ so that $\boldy=T(\Parameters)$.

The goal of this paper is to investigate under which conditions the reconfigurable \ac{EM} device is fully controllable, i.e., for a given observation space $\mathcal{O}$ with $M$ \ac{DoF} and for any observed field expressed by $\boldy$, there exists at least one device configuration $\Parameters$ such that $\boldy=\alpha \, T(\Parameters)$ for almost all possible fields, where $\alpha$ is a constant that can be filled through proper tuning of the transmitted power (link budget).  

To make the problem independent of $\alpha$, we define the function
\begin{equation}  \label{eq:s}
s(\boldy)=\frac{\boldy}{\| \boldy\|} : \Complex^M \backslash \{0\} \rightarrow \Sphere^{2M-1}
\end{equation} 
 where $\Sphere^{2M-1}=\left \{ \boldy \in \Complex^M : \| \boldy \|=1 \right \}$ denotes the unit sphere (direction space) in $\Complex^M \simeq \Reals^{2M}$.
Mathematically, the above-mentioned goal corresponds to proving that the normalized map 
\begin{equation}
\mathring{T}(\Parameters) = s(T(\Parameters))=\frac{T(\Parameters)}{\|T(\Parameters)\|} : \Reals^N \to \Sphere^{2M-1}
\end{equation}
is \emph{surjective almost everywhere} onto the full direction space. Specifically, surjectivity almost everywhere means that
\begin{equation} \label{eq:sure}
\forall \mathring{\boldy}=s(\boldy) \in \Sphere^{2M-1}, \forall \epsilon>0, \, \exists \Parameters : \left \| \mathring{T}(\Parameters) - \mathring{\boldy} \right \| < \epsilon
\end{equation}
that is, we obtain what we call \emph{dense controllability} of the \ac{EM} device's response; in other words, each direction is reachable arbitrarily closely.

\subsection{Dense Controllability Theorem} 

Here we enunciate the main Theorem of this paper, leaving the detailed proof in Sec.  \ref{Sec:proof}. 
To this purpose, define with $\JT(\Parameters)=\left [ \frac{\partial \Real{T(\Parameters)}}{\partial \Parameters} ;   \frac{\partial \Imag{T(\Parameters)}}{\partial \Parameters} \right ] \in \Reals^{2M\times N}$ the real Jacobian of $T(\Parameters)$.
%
%
The real Jacobian of the normalized map $\mathring{T}(\Parameters)$ is
\begin{equation} \label{eq:proj}
\JtT(\Parameters)=\frac{1}{\|T(\Parameters)\|} \left (\boldI_{2M} - \mathring{T}(\Parameters)\mathring{T}(\Parameters)\transpose \right ) \, \JT(\Parameters)\, 
\end{equation}
where  $\boldI_{2M} - \mathring{T}(\Parameters)\mathring{T}(\Parameters)\transpose$ projects onto the tangent space 
\begin{equation}
\Sphere_{\text{t}}= \left \{\boldv \in \Reals^{2M} : \boldv\transpose \mathring{T}(\Parameters) = 0 \right \}.
\end{equation}

\begin{theorem}[Dense controllability of $\mathring{T}$] 
\label{thm:main1}
Let $T:\Reals^N \to \Complex^M \simeq \Reals^{2M}$ be real analytic with $N \ge 2M$ and $T(\Parameters)\neq0$ everywhere.
If there exists $\Parameters_0$ such that
\begin{equation}
\rank{\boldJ_{T}(\Parameters_0)}=2M \, ,
\end{equation}
then $\mathring{T}(\Parameters)=s(T(\Parameters))$ is surjective almost everywhere on $\Sphere^{2M-1}$, i.e., 
for almost every $\mathring{\boldy} \in \Sphere^{2M-1}$ (i.e., up to a measure-zero set), there exists $\Parameters \in \Reals^N$ such that $\mathring{T}(\Parameters)=\mathring{\boldy}$. 
\end{theorem}
\begin{proof}
See Sec. \ref{Sec:proof1}.
\end{proof}

The result in Theorem \ref{thm:main1} is quite general and it can be applied to several problems in wave-domain processing,  such as \ac{DSA}, \acp{SIM}, \acp{RIS},  \acp{ESPAR}, as well as analog computing using microwave networks \cite{NerCle:25a}.

\begin{figure}[!t]
\centering\includegraphics[width=1\columnwidth]{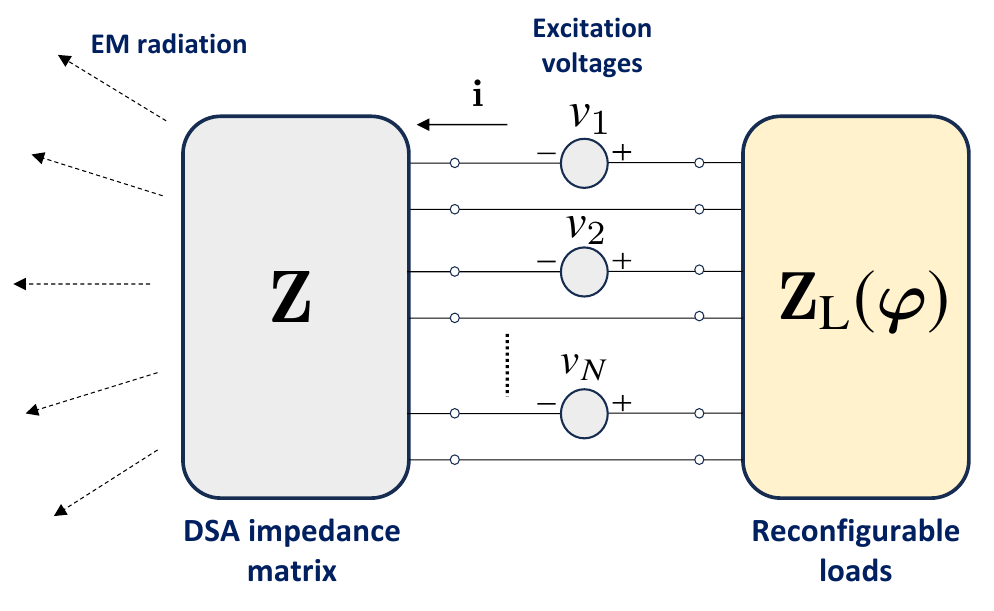}
\caption{Equivalent $N$-port circuit of the \ac{DSA}.} 
\label{Fig:NPort}
\end{figure}

\section{Controllability of DSA}
\label{Sec:Model}

In this section, we derive a specialized follow-up theorem for single-input transmitting \acp{DSA} to investigate their controllability.

\subsection{$N$-port Circuit Model for the DSA}

A transmitting \ac{DSA} is composed of $N$ antenna elements located at positions $\boldr_n \in \Reals^3$ that interact within their reactive near field region, each loaded with tunable reactances $\jmath \Parameter_n$, $n=1, 2, \ldots, N$. Denote with $a$ the radius of the sphere enclosing the \ac{DSA}. 
The first $S$ elements are active (sources), i.e., they are connected to  \ac{RF} chains and carry the information, the remaining $\Ns=N-S$ elements act as reconfigurable passive scatterers.

The \ac{DSA} can be modeled as a linear  $N$-port network as shown in Fig. \ref{Fig:NPort} \cite{Dar:J26}. 
In a compact description, we define  $\boldi=[i_1, i_2, \ldots , i_N]\transpose \in \Complex^{N \times 1}$  and $\boldv=[v_1, v_2, \ldots , v_N]\transpose \in \Complex^{N \times 1}$ the complex current and excitation voltage envelopes  (in the following denoted simply as currents and voltages) at the ports of the $N$ antennas. 
Since the $\Ns$ scatterers are passive, then the last $\Ns$ elements of $\boldv$ are zero. 
Notably, the presented framework and subsequent results are more general than \acp{DSA} where the excitation is limited to the first antenna. In this broader context, $\boldv$ can account for external excitations, such as a wave from a remote source impinging a \ac{RIS}.

All the interactions between the elements of the \ac{DSA} are captured by the impedance matrix $\boldZ  \in \Complex^{N \times N}$, which does not depend on the reconfigurable loads, and relates the voltages and currents of the $N$ ports \cite{BalB:16}. It can be computed or obtained through full-wave \ac{EM} simulations once during the design of the \ac{DSA} according to the technology adopted. 
The impedance matrix can be decomposed into the real and imaginary parts as $\boldZ = \boldR + \jmath \boldX$, and it is complex symmetric. The real matrix $\boldR \succeq 0$ is symmetric semidefinite positive, and is responsible for the radiation field, whereas $\boldX$ accounts for the reactive field. 

From basic circuit theory considerations (see Fig. \ref{Fig:NPort}), the following relationship governs the network
\begin{equation} 
\label{eq:iv}
\left ( \Zl(\Parameters) + \boldZ \right )\, \boldi=\boldv
\end{equation}
where $\Zl(\Parameters)=\jmath\,  \diag{ \Parameters} + \bRd$, being $\Parameters=[\Parameter_1, \Parameter_2, \ldots, \Parameter_N]$ collecting the configuration parameters, and $\bRd=\Rd \boldI_N$  accounts for any loss in the loads and/or the antenna with $\Rd$ denoting the loss resistance. 

For a given parameter configuration $\Parameters$ and from \eqref{eq:iv}, the currents in the system are $\boldi=\left ( \Zl(\Parameters) + \boldZ \right )^{-1} \, \boldv \, x$. The input $x$ can affect one or more active antennas depending on $\boldv$. In case of single-input \ac{DSA}, i.e.,  $S=1$, then $\boldv=\bolde_1$.  
The radiated and dissipated powers are given respectively by
\begin{equation}
\Prad=\EX{\boldi\ctranspose \boldR \boldi} \quad \quad \quad \Pd=\EX{\boldi\ctranspose \bRd \boldi} 
\end{equation}
corresponding to power efficiency $\etad=\Prad/(\Prad+\Pd)$.

Within the observation space, the electrical field component $\Em_n(\boldr)$ observed at location $\boldr$ generated by element $n$, excited by the current $i_n$, can be expressed with the series expansion 
\begin{equation} \label{eq:expansion}
	\Em_n(\boldr)=\sum_{m=1}^M  f_{m,n}  \, \Psim_m(\boldr) \, i_n \, .
\end{equation}
The total field in the observation space is given by \eqref{eq:TotE}, where $y_m=\sum_{n=1}^N f_{m,n} \, i_n$. In compact form, $\boldy=\boldF \, \boldi$, having defined $\boldF=\{ f_{m,n} \} \in \Complex^{M\times N}$. 
The matrix $\boldF$ defines the forward linear operator from dipole currents to mode's coefficients, as illustrated in Fig. \ref{fig:Spaces}.

From \eqref{eq:iv}, the input-output relationship of the \ac{DSA} is $\boldy=T(\Parameters)\, x$, where the \ac{DSA} response is 
\begin{equation}
\label{eq:Tx}
T(\Parameters)= \boldF \left (\jmath \diag{\Parameters} + \bRd+ \boldZ \right)^{-1} \boldv \in \Complex^{M \times 1} \, .
\end{equation}
As in the previous section, without loss of generality, we will assume $x=1$ so that $\boldy=T(\Parameters)$.

\begin{figure}[!t]
\centering\includegraphics[width=1\columnwidth]{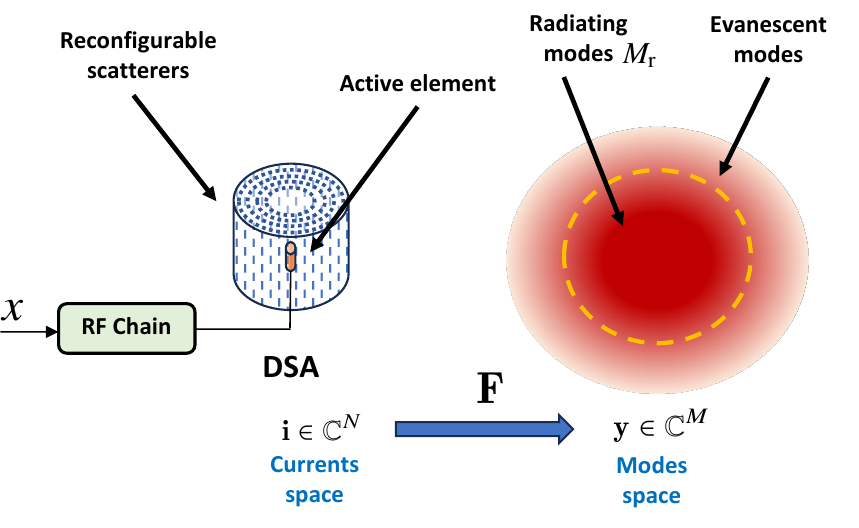}
\caption{The mapping between currents and modes in a single-input \ac{DSA}.} 
\label{fig:Spaces}
\end{figure}

\subsection{Dense Controllability Theorem for a DSA}

The following theorem specializes Theorem \ref{thm:main1} to the mapping in \eqref{eq:Tx} valid for \acp{DSA} to find sufficient conditions for dense controllability as a function of the impedance matrix $\boldZ$, the current-modes mapping function $\boldF$, i.e., on the structure and size of the \ac{DSA}. 

\begin{theorem}[Sufficient conditions for dense controllability of DSAs] 
\label{thm:main2}
Given the specific mapping $T(\Parameters)$ in \eqref{eq:Tx} and defining $\mathring{T}(\Parameters)=s(T(\Parameters))$, a sufficient condition for $\mathring{T}(\Parameters)$ to be surjective almost everywhere on $\Sphere^{2M-1}$ is given by
\begin{equation} \label{eq:suff}
    \rank{\boldF \, (\bRd+\boldZ)^{-1} \,\diag{(\bRd+\boldZ)^{-1} \, \boldv}} = M \, .
\end{equation}
\end{theorem}
\begin{proof} See Sec. \ref{Sec:proof2}.
\end{proof}

\emph{Remark}: It is interesting to note that the condition in \eqref{eq:suff} is independent of the specific configuration parameters of the \ac{DSA} $\Parameters$. 

Define the vector $\tilde{\boldv}=(\bRd+\boldZ)^{-1} \, \boldv$, the matrix $\tilde{\boldF}=\boldF \, (\bRd+\boldZ)^{-1}$, and let $\mathcal{I}=\left \{n : |\tilde{v}_n|>\gamma>0 \right \}$ be the indices of not negligible components of $\tilde{\boldv}$ for some small $\gamma$. Condition \eqref{eq:suff} requires that: 
\begin{enumerate} 

\item  The cardinality of $\mathcal{I}$ must be at least $M$: This means that the excitation vector $\boldv$ combined with the coupling effect determined by the impedance matrix $\boldZ$ must affect at least $M$ elements of the \ac{DSA}; 

\item The submatrix $\left [\tilde{\boldF} \right ]_{:, \mathcal{I}}$ must have rank $M$: There should be a coupling between the excited \ac{DSA}'s elements and the radiating modes so that all radiating modes are potentially reachable. This implies that a necessary condition for surjectivity is $\rank{\tilde{\boldF}}=M$.

\end{enumerate}

The result of Theorem \ref{thm:main2} provides a powerful and general tool to assess whether a specific structure is controllable within the observation space $\mathcal{O}$.
On the other hand, the result is useful in determining the design guidelines for a \ac{DSA} in terms of geometry, element density/spacing, excitation, etc.

For instance, consider a \ac{DSA} of $N$ elements where only the first element is actively driven, i.e., $\boldv=\bolde_1$, while the remaining $N-1$ elements are passively loaded with tunable reactances. Suppose that the elements are uncoupled so that $\boldZ=\Rr \boldI$, where $\Rr$ denotes the radiation resistance of the single antenna element. It follows that only one element of $\tilde{\boldv}=(\Rr+\Rd)^{-1} \boldv$ is different from zero, and condition 1) is not satisfied. Therefore, such a \ac{DSA} is not fully controllable.\footnote{It can be easily shown that the Jacobian has rank 1 for all $\Parameters$ in this case.}  This highlights the fundamental role played by coupling in \acp{DSA} for full controllability.
On the contrary, if all elements of the antenna structure are excited by an external impinging wave (all elements of $\boldv$ are different from zero), full controllability might potentially be achieved even in the case of uncoupled elements. This is a typical situation where the structure is a diagonal \ac{RIS} with uncoupled elements (e.g., element spacing equal to $\lambda/2$).  
As a final remark, $N\ge 2M$ is a necessary condition for controllability, i.e., the number of elements should be at least twice the number of modes considered to represent the field in the observation space.

\section{Proofs of Theorems \ref{thm:main1} and \ref{thm:main2}}
\label{Sec:proof}

\subsection{Proof of Theorem \ref{thm:main1}}
\label{Sec:proof1}

Here we recall some definitions, theorems, and lemmas of differential geometry \cite{Lee:B13} that will be used to prove Theorem \ref{thm:main1}.

\begin{definition}[Regular and critical points]
Let $\Ncal$ and $\Mcal$ be differentiable manifolds and let $f: \Ncal \to \Mcal$ be a differentiable map between them. The map $f$ is a submersion at a point $\boldp \in \Ncal$ if its differential is a surjective linear map, i.e., its Jacobian has full rank in $\boldp$. 
In this case, $\boldp$ is called a \emph{regular point} of the map $f$ and $f(\boldp)$ \emph{regular value},  otherwise, $\boldp$ is a \emph{critical point} and $f(\boldp)$ is a \emph{critical value}.
\end{definition}


\begin{theorem}[Submersion theorem]
\label{thm:sub}
A standard result in differential geometry states that submersions are locally open maps. In particular, if $f$ is a submersion at $\boldp$, then $f$ maps a neighborhood of $\boldp$ in $\Ncal$ to an open set in $\Mcal$ \cite{Lan:B99}.
\end{theorem}

\begin{theorem}[Sard's theorem]
\label{thm:Sard}
Given the smooth map $f: \Ncal \to \Mcal$ and the set $\mathcal{C}$ of its critical points,\footnote{A smooth map has continuous derivatives of all orders. A real analytical function is a smooth map.} the image $f(\mathcal{C})$ of $\mathcal{C}$ (set of critical values) has Lebesgue measure zero \cite{Sar:48}.
\end{theorem}

We now apply the above definitions and Theorems to our scope. 

\begin{lemma}[Rank preservation almost everywhere]
\label{lem:rank}
Let $T : \mathbb{R}^N \to \mathbb{C}^M$ be a generic real analytic map, and let 
$\JT(\Parameters) \in \mathbb{R}^{2M \times N}$ denote its real Jacobian 
under the identification $\mathbb{C}^M \simeq \mathbb{R}^{2M}$.
If there exists $\Parameters_0 \in \mathbb{R}^N$ such that $\rank{\JT(\Parameters_0)} = 2M$,  then $\rank{\JT(\Parameters)} = 2M$ for almost every $\Parameters \in \mathbb{R}^N$.
\end{lemma}

\begin{proof}
Since $\rank{\JT(\Parameters_0)} = 2M$, there exists at least one  $2M \times 2M$ minor of $\JT(\Parameters)$ whose determinant does not vanish at $\Parameters_0$.  Let $q(\Parameters) = \det{\left [\JT(\Parameters)\right ]_{\mathcal{I},\mathcal{J}}}$ denote such a minor.
Because $T$ is real analytic, all entries of $\JT(\Parameters)$ are real analytic functions of parameters $\Parameters$, and therefore $q(\Parameters)$ is real analytic because the determinant is a polynomial in these entries. Moreover, $q(\Parameters_0) \neq 0$, hence $q(\Parameters)$ is not identically zero. A standard result in real analytic geometry states that the zero set of a nontrivial real analytic function has Lebesgue measure zero \cite{Lee:B13}.  Therefore,
\begin{equation}
\mathcal{S} = \{\Parameters \in \mathbb{R}^N : \rank{\JT(\Parameters)} < 2M\}
\subseteq \{\Parameters : q(\Parameters) = 0\}
\end{equation}
has Lebesgue measure zero. This proves the claim.
\end{proof}

\begin{lemma}[]
\label{lem:sub}
The map $\mathring{T}(\Parameters)$ is a submersion onto the sphere at almost every $\Parameters$.
\end{lemma}
\begin{proof}
From Lemma \ref{lem:rank}, $\rank{\JT(\Parameters)}=2M$ at almost every $\Parameters$. The projection in \eqref{eq:proj} reduces the rank by at most one (along the radial direction), giving $\rank{\JtT(\Parameters)}=2M-1$ at almost every $\Parameters$. 
\end{proof}

\begin{figure}[!t]
\centering 
\centering\includegraphics[width=1\columnwidth]{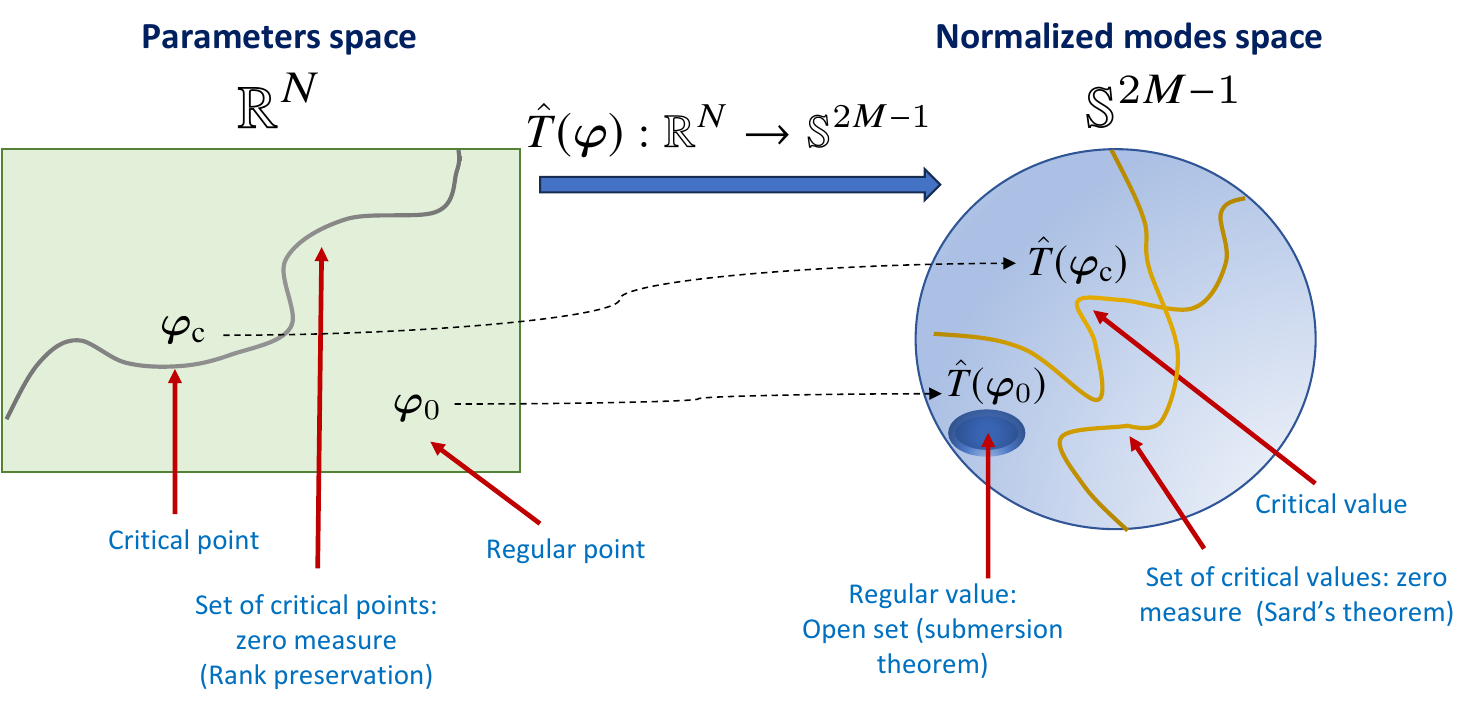}
\caption{Graphical sketch of the proof of Theorem \ref{thm:main1}.} 
\label{Fig:Theorem}
\end{figure}

We are now in the position to prove Theorem \ref{thm:main1}, whose graphical sketch is summarized in Fig. \ref{Fig:Theorem}. 
  
\begin{proof}[Proof of Theorem \ref{thm:main1}] 
From the hypothesis of Theorem \ref{thm:main1}, $T(\Parameters)$ is real analytic and there exists at least $\Parameters_0$ such that $\rank{\boldJ_{T}(\Parameters_0)}=2M$. From Lemma \ref{lem:rank}, this property holds for all $\Parameters$ outside a set of measure zero, i.e., the Jacobian is full rank almost everywhere. 
From the submersion Theorem \ref{thm:sub}, if $\mathring{T}$ is a submersion at $\Parameters$, then $\mathring{T}$ maps a neighborhood of $\Parameters$ in $\Reals^N$ to an open set in $\Sphere^{2M-1}$ (local openness). Moreover, from Lemma \ref{lem:sub}, $\rank{\JtT(\Parameters)}=2M-1$ at almost every $\Parameters$.

The Sard's theorem \ref{thm:Sard} states that the set of critical values has measure zero in $\Sphere^{2M-1}$. Thus, almost every $\mathring{\boldy} \in \Sphere^{2M-1}$ is a regular value and the image of $\mathring{T}$ contains an open set in $\Sphere^{2M-1}$ almost everywhere. 
Let $\mathcal{U} \subset \image{\mathring{T}}$ denotes the union of all open neighborhoods corresponding to regular values. Then $\mathcal{U}$ is an open subset of $\Sphere^{2M-1}$, and its complement has measure zero.
In fact, by the definition of closure,
\begin{equation} \label{eq:closure}
\overline{\image{\mathring{T}}} = \{ \mathring{\boldy} \in \Sphere^{2M-1} : \exists\, \mathring{\boldy}_k \in \image{\mathring{T}},\, \mathring{\boldy}_k \to \mathring{\boldy} \} \, .
\end{equation}
Since $\mathcal{U}$ is dense in $\Sphere^{2M-1}$, for any $\mathring{\boldy} \in \Sphere^{2M-1}$ there exists a sequence $\{\mathring{\boldy}_k\} \subset \mathcal{U} \subset \image{\mathring{T}}$ such that $\mathring{\boldy}_k \to \mathring{\boldy}$. Therefore, $\mathring{\boldy} \in \overline{\image{\mathring{T}}}$ and because $\mathring{\boldy}$ was arbitrary, $\overline{\image{\mathring{T}}} = \Sphere^{2M-1}$. 
In other words, every $\mathring{\boldy} \in \Sphere^{2M-1}$ can be approached with any desired accuracy, i.e.,    
\eqref{eq:sure} holds.
\end{proof}

\emph{Remark}:  
Exact surjectivity of $\mathring{T}$ is not guaranteed, since some measure-zero critical values may never be attained exactly. However, this result ensures that the image is dense on the sphere. We obtain a dense controllability of the radiation pattern, i.e., each point is reachable arbitrarily closely, which is sufficient for practical applications such as arbitrary radiation pattern synthesis.

\subsection{Proof of Theorem \ref{thm:main2}}
\label{Sec:proof2}

Before proceeding with the proof of Theorem \ref{thm:main2}, consider the following lemma:

\begin{lemma}
\label{thm:invertibility}
For all $\Parameters \in \Reals^N$, the function $\boldA(\Parameters)=\jmath\, \diag{\Parameters} + \bRd+ \boldZ$ is invertible.
\end{lemma}
\begin{proof}
For convenience, define $\bRt=\boldR+\bRd$ so that $\boldA(\Parameters)=\bRt + \jmath\, (\diag{\Parameters} +\boldX)$. By contracdiction, assume that $\boldA(\Parameters) \, \boldu = 0$ for some $\boldu \in \Complex^N \setminus \{0\}$, i.e., it is not invertible.  Since  $\bRt$ and $(\diag{\Parameters}+  \boldX)$ are real symmetric and hence Hermitian matrices, we have $\boldu\ctranspose \bRt \boldu \in \Reals$ and $\boldu\ctranspose (\diag{\Parameters}+  \boldX) \boldu \in \Reals$, then $\Re(\boldu\ctranspose \boldA(\Parameters) \boldu) = \boldu\ctranspose \bRt \boldu>0$ for all $\Parameters \in \Reals^N$ because $\bRt \succ 0$, but this is in contradiction with the assumption, then $\boldA(\Parameters)$ is invertible.
\end{proof}

\begin{corollary}
The function  $T(\Parameters)=\boldF \, \boldA(\Parameters) \, \boldv$ is real analytic, and hence all partial derivatives exist and are continuous.
\end{corollary}
\begin{proof}
From Lemma \ref{thm:invertibility}, $\boldA(\Parameters)$ is invertible for all $\Parameters \in \Reals^N$ and the elements of $T(\Parameters)$ are rational functions of $\Parameters$.
\end{proof}

In preparation for the proof of Theorem \ref{thm:main2}, we compute the complex Jacobian of $T(\Parameters)$. 
Using the matrix derivative identity
\begin{equation}
\frac{\partial \boldA(\Parameters)^{-1}}{\partial \Parameter_n} = - \boldA(\Parameters)^{-1} \frac{\partial \boldA(\Parameters)}{\partial \Parameter_n} \boldA(\Parameters)^{-1}, 
\quad 
\frac{\partial \boldA(\Parameters)}{\partial \Parameter_n} = \jmath \boldE_n 
\end{equation}

we get
\begin{align}
    \frac{\partial T(\Parameters)}{\partial \Parameter_n} &= - \jmath \boldF \boldA(\Parameters)^{-1} \boldE_n \boldA(\Parameters)^{-1}  \boldv  \nonumber \\
    &= - \jmath \left [\boldA(\Parameters)^{-1} \boldv \right ]_n \boldF \left [\boldA(\Parameters)^{-1}\right ]_{:,n} \, .
\end{align}

Thus, the complex Jacobian is
\begin{equation} \label{eq:JT}
\JTc(\Parameters) = -\jmath \boldF \boldA(\Parameters)^{-1} \diag{\boldA(\Parameters)^{-1} \boldv}  \, .
\end{equation}

\begin{proof}[Proof of Theorem \ref{thm:main2}] 
Since the mapping in \eqref{eq:Tx} is real analytic and $\rank{\JT(\Parameters)}=2\, \rank{\JTc(\Parameters)}$, to satisfy the hypothesis of Theorem \ref{thm:main1}, it is sufficient to show that its (complex or real) Jacobian has full rank for $\Parameters_0=\boldzero$. In this case, $\boldA(\boldzero)=\bRd+\boldZ$ and the Jacobian in \eqref{eq:JT} becomes
\begin{align} \label{eq:JT1}
    \JTc(\boldzero) &= - \jmath \boldF (\bRd+\boldZ)^{-1} \diag{(\bRd+\boldZ)^{-1} \boldv}  
\end{align}
thus the condition \eqref{eq:suff} follows.
\end{proof}


\section{2D and 3D Unbounded Observation Spaces}
\label{Sec:Modes}
We derive the mapping matrix $\boldF$, which links the currents of the $N$ antenna elements to the modes, for 2 cases of large interest: 3D and 2D structures observed in unbounded space. 

\subsection{3D Spherical Radiation Mapping Matrix $\boldF$}
\label{Sec:3D}

The electrical field outside the sphere of radius $a$ containing the \ac{DSA} can be described through the outgoing vector spherical waves expansion, which includes two sets of basis functions, respectively, for the electric (TE) and magnetic (TM) modes \cite{HarringtonBook:2001,GusJeelSchCap:22} as
\begin{equation}
\Em(\boldr)= \sum_{\ell=1}^{\infty} \sum_{q=-\ell}^{\ell} 
\left[ a_{\ell,q} \, \EMsymb{N}_{\ell,q}^{(3)}(\kappa \boldr) 
+ b_{\ell,q} \, \EMsymb{M}_{\ell,q}^{(3)}(\kappa \boldr ) \right]
\end{equation}
where $\kappa=2\pi/\lambda$ is the wave number and $\EMsymb{N}_{\ell,q}^{(3)}$ and $\EMsymb{M}_{\ell,q}^{(3)}$ are the TM and magnetic TE vector outgoing spherical waves, respectively.

We simplify the scenario by considering the \ac{DSA} to be composed of radial-oriented infinitesimal dipoles. This is a typical approach used in the literature to obtain high-level insight about the behavior of the system without the complexity introduced by polarization \cite{IvrNos:10}. The analysis comprising polarization effects is left for future work.
As a consequence, the vector spherical expansion collapses to a scalar spherical harmonic expansion, because the field is purely TM. 
In particular, for a single dipole of length $\Delta$ at position $\boldr_n$, the vector spherical wave coefficients are
\begin{equation}
a_{\ell,q}^{(n)} =\frac{\kappa \eta \Delta \ell (\ell +1)}{4 \pi}    \, j_\ell(\kappa r_n) 
\, Y_{\ell,q}^*(\hat{\boldr}_n) \, , \qquad
b_{\ell,q}^{(n)} = 0
\end{equation}
where $r_n = \| \boldr_n\|$, $j_\ell(\cdot)$ is the spherical Bessel function of order $l$, and $Y_{\ell,q}(\hat{\boldr}_n)$ is the scalar spherical harmonic \cite{HarringtonBook:2001}. Therefore, it is $a_{\ell,q}=\sum_{n=1}^N a_{\ell,q}^{(n)}$. 

Truncating the outer sum to a certain value $L_\mathrm{max}$, the actual number of considered modes is $M=(L_\mathrm{max}+1)^2$ corresponding to a maximum directivity $D \simeq M$ \cite{Har:58}.
From the Chu-Harrington limit, the number $M_{\text{r}}$ of dominant radiating spherical modes supported by a region of radius $a$ satisfies \cite{Har:58,KilMarMac:17}
\begin{equation}
    M_{\text{r}}=M_{\text{3D}} \approx (\kappa a)^2 \, 
\end{equation}
therefore it must be $M\ge M_{\text{r}}$. Note that $D=M_{\text{r}} \approx (\kappa a)^2=\frac{4 \pi}{\lambda^2} A$ is the maximum available directivity for a non-super directive aperture of area $A=\pi a^2$ \cite{KilMarMac:17}. 
If $M>M_{\text{r}}$, then also evanescent modes are included, and superdirective patterns can be obtained owing to the higher number of available \ac{DoF} at the expense of higher $Q$ factors and induced currents at the scatterers.

The corresponding mapping matrix $\boldF \in \Complex^{M \times N}$ in the expansion \eqref{eq:expansion} has elements $f_{m, n} = a_{\ell,q}^{(n)} $ and $\Psim_m(\boldr)=\EMsymb{N}_{\ell,q}^{(3)}(\kappa \boldr)$, with $m=2(\ell-1)+q + \ell +2$, and $\ell=1,2, \dots,L_\mathrm{max}$, and $\quad q=-\ell,\dots, 0, \ldots, \ell$. 
Each row of $\boldF$ is associated with a spherical harmonic mode, and each column with an antenna element. The rank of $\boldF$ determines the number of independent radiated \ac{DoF}.
It is interesting to note that the radiated power can be written as
\begin{equation}
\Prad=\EX{\boldi\ctranspose \boldR \boldi}= \frac{1}{\eta \kappa^2}\EX{\boldy\ctranspose \boldy}= \frac{1}{\eta \kappa^2}\EX{\boldi\ctranspose \boldF\ctranspose \boldF \boldi}
\end{equation}
from which the following relationship between $\boldR$ and $\boldF$ can be inferred: $\boldR= \frac{1}{\eta \kappa^2} \boldF\ctranspose \boldF$.
The null space of $\boldF$ corresponds to the reactive modes (not radiating).
In general, it is  $\boldZ=\frac{1}{\eta \kappa^2} \boldF\ctranspose \boldF  + \boldZ_{\text{res}}$, where $\boldZ_{\text{res}}$ contains reactive contributions.


\begin{figure}[!t]
\centering\includegraphics[width=0.6\columnwidth,trim={0 0 0 2cm},clip]{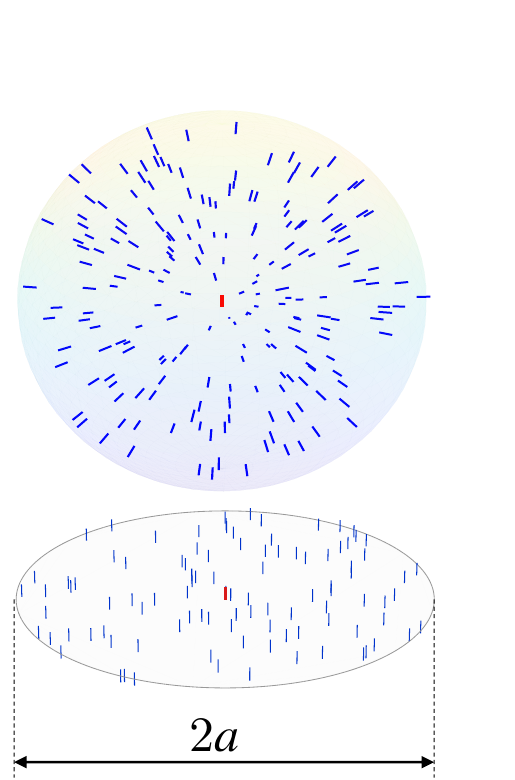}
\caption{Randomly distributed dipoles within a sphere (top) and a disk (bottom).} 
\label{fig:RandomDipoles}
\end{figure}

\subsection{2D Disk Radiation Mapping Matrix $\boldF$}
\label{Sec:2D}

For dipoles distributed on a disk of radius $a$ in the $x$-$y$ plane, the distance-normalized far field on the plane is proportional to 
\begin{equation}
    \Em(\phi) \propto \sum_{n=1}^N i_n \, e^{\jmath \kappa  r_n \cos(\phi-\phi_n)}  \
\end{equation}
with  $\phi_n = \angle \boldr_n$.
A basis set for the beam space domain is given by the Fourier series in the angular modes $\left \{ e^{\jmath m \phi} \right \}$. 
Using the Fourier-Bessel expansion, we have
\begin{equation}
    e^{\jmath \kappa  r_n \cos(\phi-\phi_n)} = \sum_{\ell=-\infty}^{\infty} \jmath^\ell J_\ell(\kappa r_n) \, e^{\jmath \ell (\phi-\phi_n)} 
\end{equation}
 where $J_\ell(\cdot)$ represents the $\ell$th order Bessel function of the first kind. By truncating the series to the value $L_{\text{max}}$, the elements of the mapping matrix $\boldF$ are
\begin{equation}
    f_{m,n} = \jmath^\ell J_{\ell}(\kappa r_n) \, e^{-\jmath \ell \phi_n}, \quad m=1, 2, \dots, M 
\end{equation}
where $\ell=m-L_{\text{max}}-1$,  giving a total of $M= 2 L_{\text{max}}+1$ modes.
In this case, since $J_{\ell}(x) \approx 0$ for $|\ell|>x$, the number of well-radiating modes is approximately given by $M_{\text{r}}=M_\mathrm{2D}  \approx 2 \kappa a+1$ and the directivity is $D \approx M \, D_{\text{vertical}}$, being $D_{\text{vertical}}$ the directivity along the vertical direction.

It is worth noticing that for randomly distributed dipoles the condition $N\ge 2M$, with $M$ taken equal to $M_{\text{2D}}$ or $M_{\text{3D}}$, leads to an average nearest neighbor distance between elements $\lesssim  \frac{\lambda}{4}\left (\frac{a}{\lambda} \right )^{\frac{1}{d}}$, where $d=2$ and $d=3$ for the 2D and 3D cases, respectively. This implies that for a small \ac{DSA} ($a \approx \lambda$), the average minimum distance between elements should be $ \lesssim \lambda/4$. In contrast, a larger spacing (smaller density) can be used for a larger \ac{DSA}, provided that the minimum cardinality condition of the set $\mathcal{I}$, associated with the vector $\tilde{\boldv}$ defined in Sec. \ref{Sec:proof2}, is satisfied (source illumination condition). Interestingly, $\lambda/4$ was also the optimal spacing found through numerical simulations in \cite{Dar:J26}.

\begin{figure}[!t]
\centering\includegraphics[width=1\columnwidth]{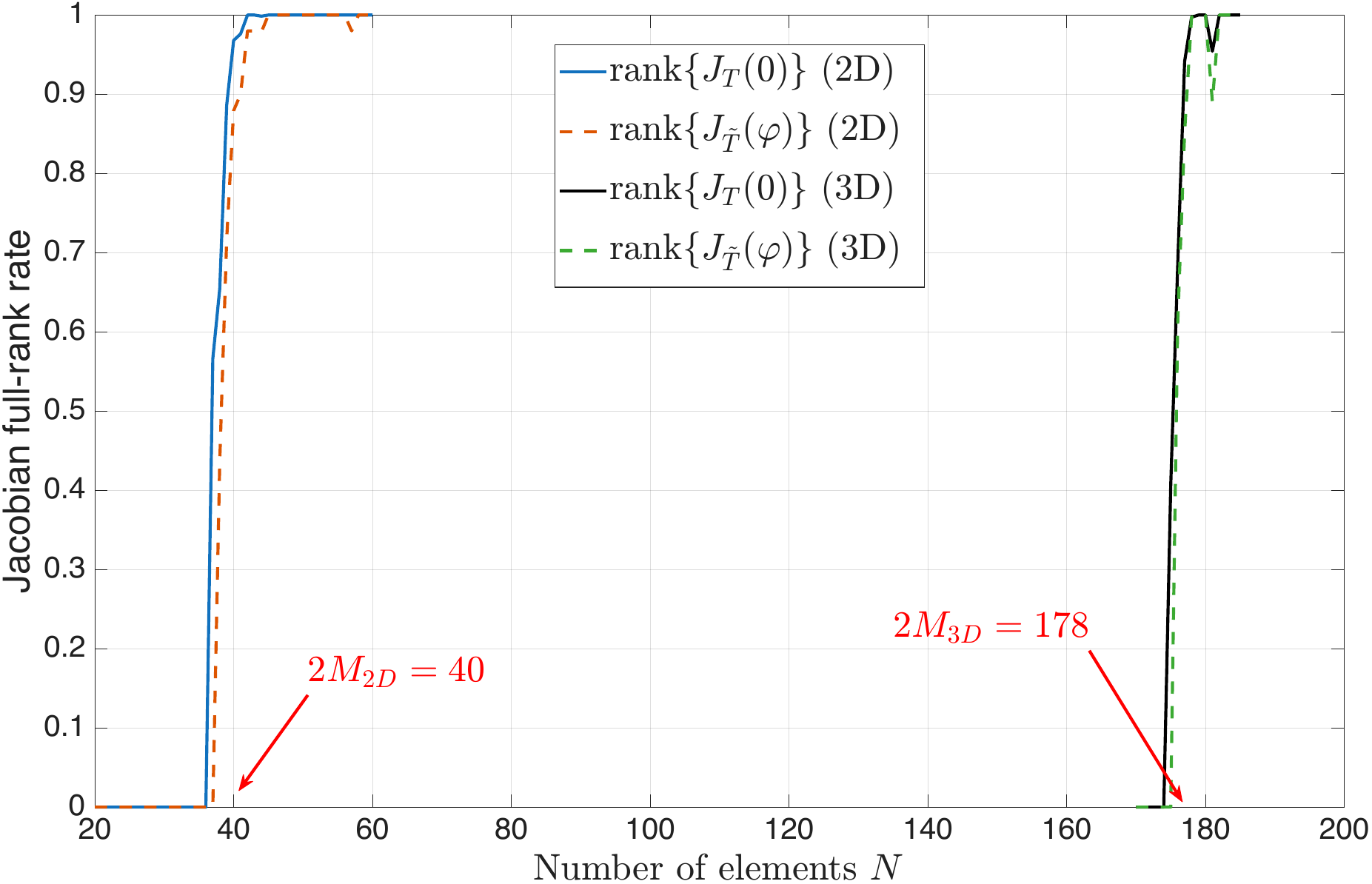}
\caption{Numerical validation of Lemma \ref{lem:rank}. } 
\label{fig:Jacobian}
\end{figure}

\section{Numerical Results}
\label{Sec:NumericalResults}

In this section, we present some numerical results to validate the theoretical framework illustrated in the previous sections. 

As an intermediate step, we begin by numerically verifying the validity of Lemma \ref{lem:rank} by examining the correlation between condition \eqref{eq:suff} (linked to \eqref{eq:JT1}), which is independent of the parameter set $\Parameters$, and the rank of $\JtT(\Parameters)$. This analysis is carried out on a large number of randomly generated parameter sets (400) and geometries (50), and computing the rate at which the rank is equal to $2M$. The rank is computed using a relative tolerance of $10^{-5}$.
Specifically, we consider $N$ Hertzian dipoles randomly distributed within a circle of radius $a = 1.5\lambda$ (2D case) and within a sphere of the same radius (3D case) (see Fig. \ref{fig:RandomDipoles}). The matrix $\boldF$ is computed using the series expansions described in Sec. \ref{Sec:2D} and Sec. \ref{Sec:3D}, respectively. We set $\Rd = 0.1\,\Omega$, and the impedance matrix $\boldZ$ is evaluated analytically as in \cite{BalB:16}. The active element is placed at the center.
The corresponding number of well-radiating modes is $M = M_{\text{r}} = M_{\mathrm{2D}} \approx 2\kappa a + 1 = 20$ in the 2D case, and $M = M_{\text{r}} = M_{\mathrm{3D}} \approx (\kappa a)^2 = 89$ in the 3D case.

Fig. \ref{fig:Jacobian} shows the full rank rate of $\JtT(\Parameters)$, conditioned to \eqref{eq:suff}, and the full rank rate of \eqref{eq:suff}. As predicted by Lemma \ref{lem:rank}, a strong agreement can be noticed. Moreover, we observe an abrupt change in slope around $N \approx 2M$, from which we can conjecture that, for random geometries and with high probability, the condition $N \geq 2M$ is not only necessary but also sufficient to ensure \eqref{eq:suff}. Consequently, by Theorem \ref{thm:main2}, $\hat{T}(\Parameters)$ is almost everywhere surjective, i.e., the \ac{DSA} is densely controllable. 

The same experiment, repeated with $N$ elements uniformly distributed over a circle of radius $a$, corresponding to a typical \ac{ESPAR} configuration, yields a full-rank rate equal to zero over the considered range of $N$. Therefore, no conclusions can be drawn regarding the dense controllability of this configuration. However, the following numerical examples reveal the presence of an approximation error floor, suggesting that the system is not densely controllable.

\begin{figure}[!t]
\centering\includegraphics[width=1\columnwidth]{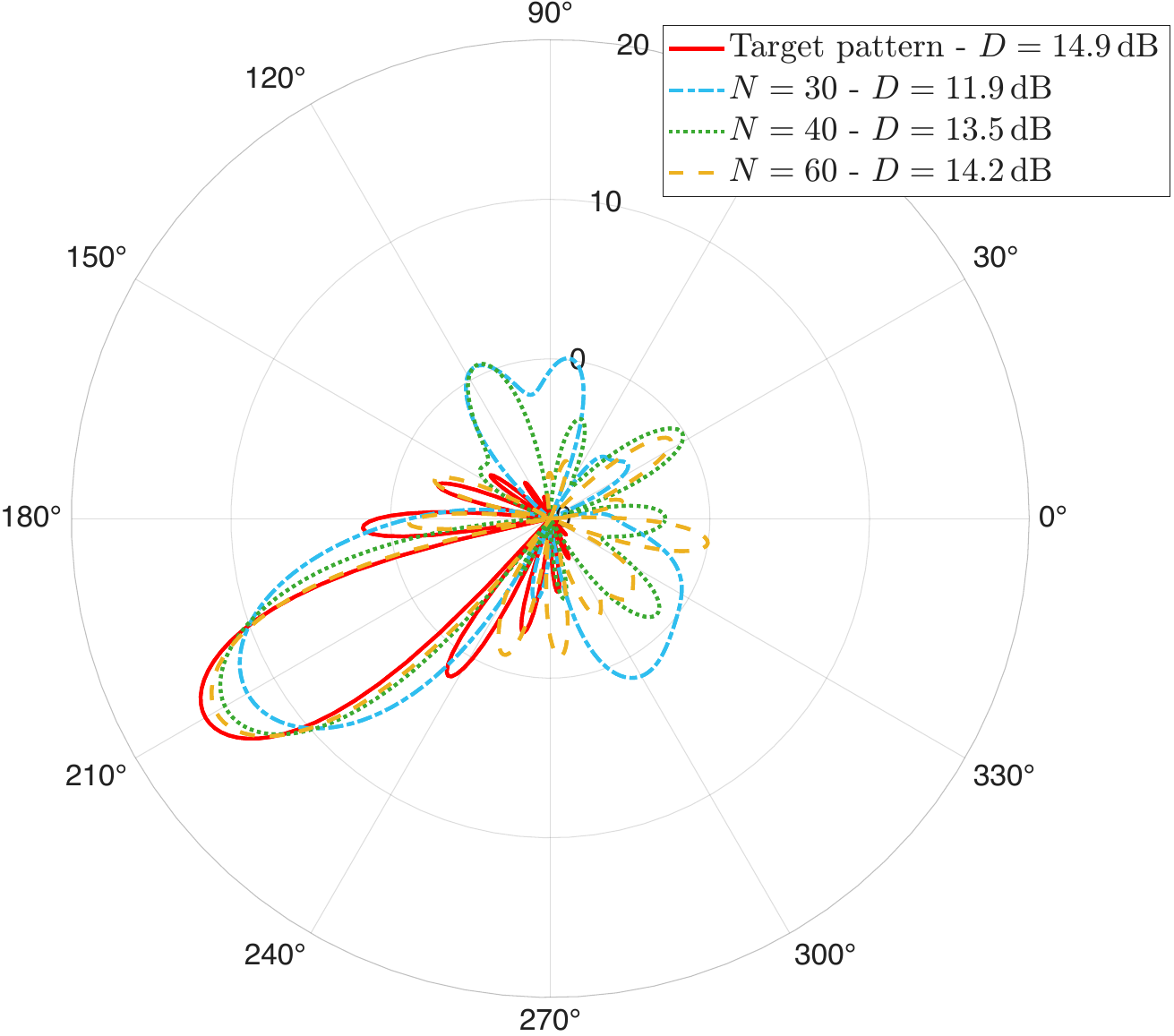}
\caption{Radiation pattern and directivity values $D$ as a function of the number of elements for a target single beam toward angle $\phi_1=210^{\circ}$. $M=20$.} 
\label{fig:Pattern}
\end{figure}

Numerically validating surjectivity is a prohibitive task. However, useful insights can be obtained by optimizing multiple random geometric configurations of the \ac{DSA} for a given desired response and evaluating the achieved accuracy.
To this end, we consider the 2D scenario described above and fix a target response $\mathring{\boldy}$. For the $i$th geometric configuration, the following optimization problem is solved numerically using a standard quasi-Newton method with 5000 iterations:
\begin{equation} \label{eq:min}
\Parameters^{(i)}=\arg \min_{\Parameters} \frac{\| \mathring{T}(\Parameters)-\mathring{\boldy} \|^2}{M} \, .
\end{equation}

The objective function in \eqref{eq:min} corresponds to the normalized \ac{MSE} achieved with the parameter set $\Parameters$. The solution $\Parameters^{(i)}$ thus defines the $i$th configuration of the parameters of the \ac{DSA}.

Some representative examples are shown in Figs. \ref{fig:Pattern}, \ref{fig:Patternbis}, and \ref{fig:PatterSuper}, where radiation patterns are reported to provide a qualitative assessment of the pattern approximation accuracy as the number of elements $N$ increases. These results are obtained by considering a snapshot of vertically oriented half-wave dipoles randomly distributed within a circle of radius $a$, with the impedance matrix $\boldZ$ analytically evaluated according to the coupling model in \cite{BalB:16}.

\begin{figure}[!t]
\centering\includegraphics[width=0.91\columnwidth]{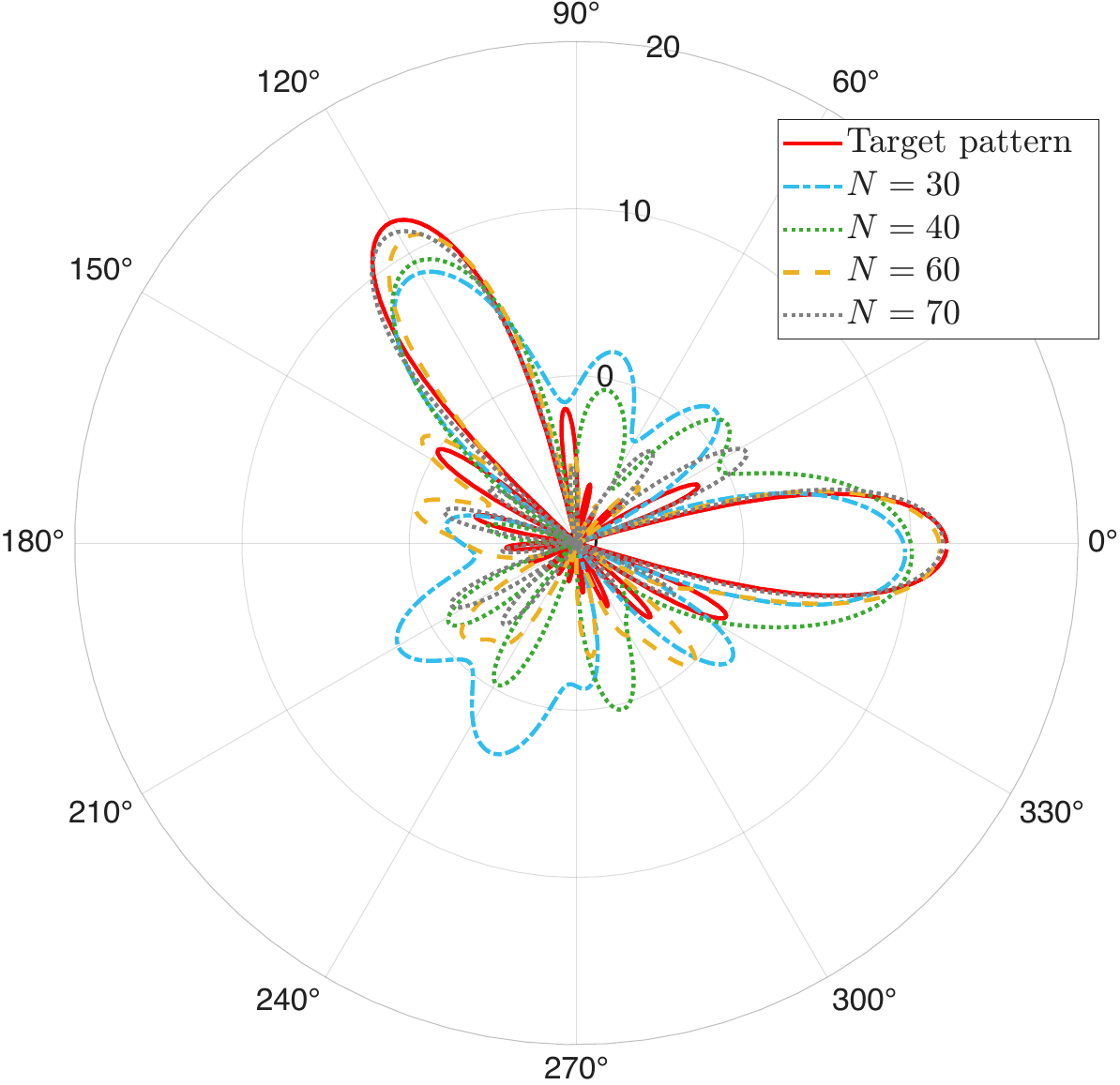}
\caption{Radiation pattern as a function of the number of elements for a target double beam toward angles $\phi_1=0^{\circ}$ and $\phi_2=120^{\circ}$. $M=20$.} 
\label{fig:Patternbis}
\end{figure}

Specifically, in Fig. \ref{fig:Pattern}, the target response corresponds to a beam pointing toward the angle $\phi_1 = 210^{\circ}$, i.e., $[\mathring{\boldy}]_m = e^{-\jmath \ell \phi_1}$, with $m = 1, 2, \ldots, M_{\text{r}}$ and $\ell = m - (M_{\text{r}}-1)/2 - 1$. The maximum available directivity is $D \approx 1.64\, M_{\text{r}}=15.1\,$dB, being $1.64$ the directivity of the half-wave length dipole and $M_{\text{r}}=M_{\text{2D}}$. 
As observed, when $N > 2M_{\text{r}} = 40$, the resulting radiation pattern closely matches the target. The power efficiency $\eta_{\text{d}}$ ranges between 0.90 and 0.99 for the values of $N$ considered.
A more elaborate double-beam target radiation pattern is considered in Fig. \ref{fig:Patternbis}, using the same setup as in the previous figure. Also in this case, the accuracy becomes satisfactory when $N > 40$.

\begin{figure}[!t]
\centering\includegraphics[width=1\columnwidth]{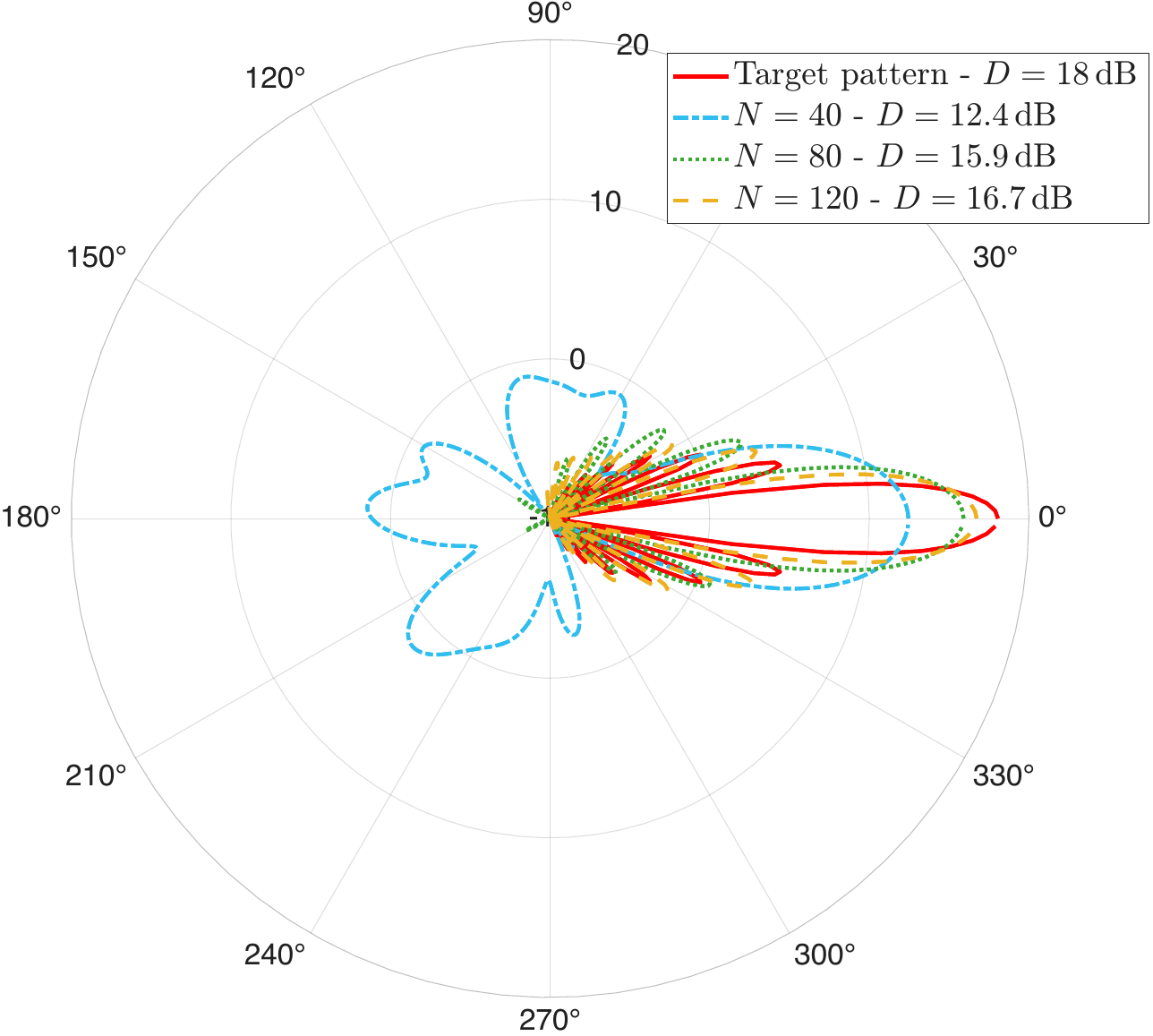}
\caption{Superdirective radiation pattern and directivity values $D$ as a function of the number of elements for a target single beam toward angle $\phi_1=0^{\circ}$. $M=40$.} 
\label{fig:PatterSuper}
\end{figure}

A particularly interesting scenario is investigated in Fig. \ref{fig:PatterSuper}, where $M = 40$, i.e., larger than $M_{\text{r}} = 20$. The target beam pointing toward the angle $\phi_1 = 0^{\circ}$ exhibits superdirective behavior, as several evanescent modes are excited (target directivity $D=18\,$dB versus $D=14.9\,$dB in Fig. \ref{fig:Pattern}). In this case, the power efficiency $\eta_{\text{d}}$ drops below 0.9 for $N > 80$, reaching values as low as 0.75. This reflects significantly increased induced currents in the scatterers' loads, which in turn lead to higher power dissipation due to $\Rd$. Larger values of $N$ are required to achieve satisfactory accuracy due to the increased value of $M$.

Finally, in Fig. \ref{fig:MSE}, the normalized \ac{MSE}, averaged over 30 random geometric configurations, is reported as a function of the number of elements $N$ for the cases considered in the previous figures. As expected, the \ac{MSE} decreases with $N$ in a random deployment.
For comparison, the case where the elements are uniformly distributed in a circle of radius $a$ is also shown. In contrast to the case where elements are randomly distributed within the circle, this configuration exhibits worse performance and the presence of a \ac{MSE} floor. These results indicate that the enhanced coupling among all elements in the random deployment plays an essential role in the controllability of the \ac{DSA} and enables potential superdirective behavior.

Further numerical examples of \ac{DSA} applied to single- and multi-user \ac{MIMO} tasks, such as \ac{MIMO} precoding at the \ac{EM} level, as well as full-wave simulation validation, can be found in \cite{Dar:J26} and \cite{BenTroMasCosDar:C26}, respectively.


\section{Conclusion}
\label{Sec:Conclusion}

This paper derives fundamental theorems for the controllability of generic reconfigurable \ac{EM} devices. A differential geometry-based theorem provides a sufficient controllability condition for linear \ac{EM} devices under mild transfer matrix assumptions that show that verifying the controllability for a single configuration ensures the full controllability of the \ac{EM} device. Its applicability to generic transfer functions makes it a broadly useful mathematical result.

For the case of single-input \acp{DSA}, we prove compliance with the theorem and derive a closed-form criterion for dense controllability depending on the dimension of the observation space, the number of elements composing the \ac{DSA}, and mutual coupling. This facilitates a rapid assessment of controllability for a given device characterized either analytically, through simulations, or by measurements.

We further discuss the interplay between the number of elements, physical size, \ac{DoF}, and directivity through numerical examples.


\begin{figure}[!t]
\centering\includegraphics[width=1\columnwidth]{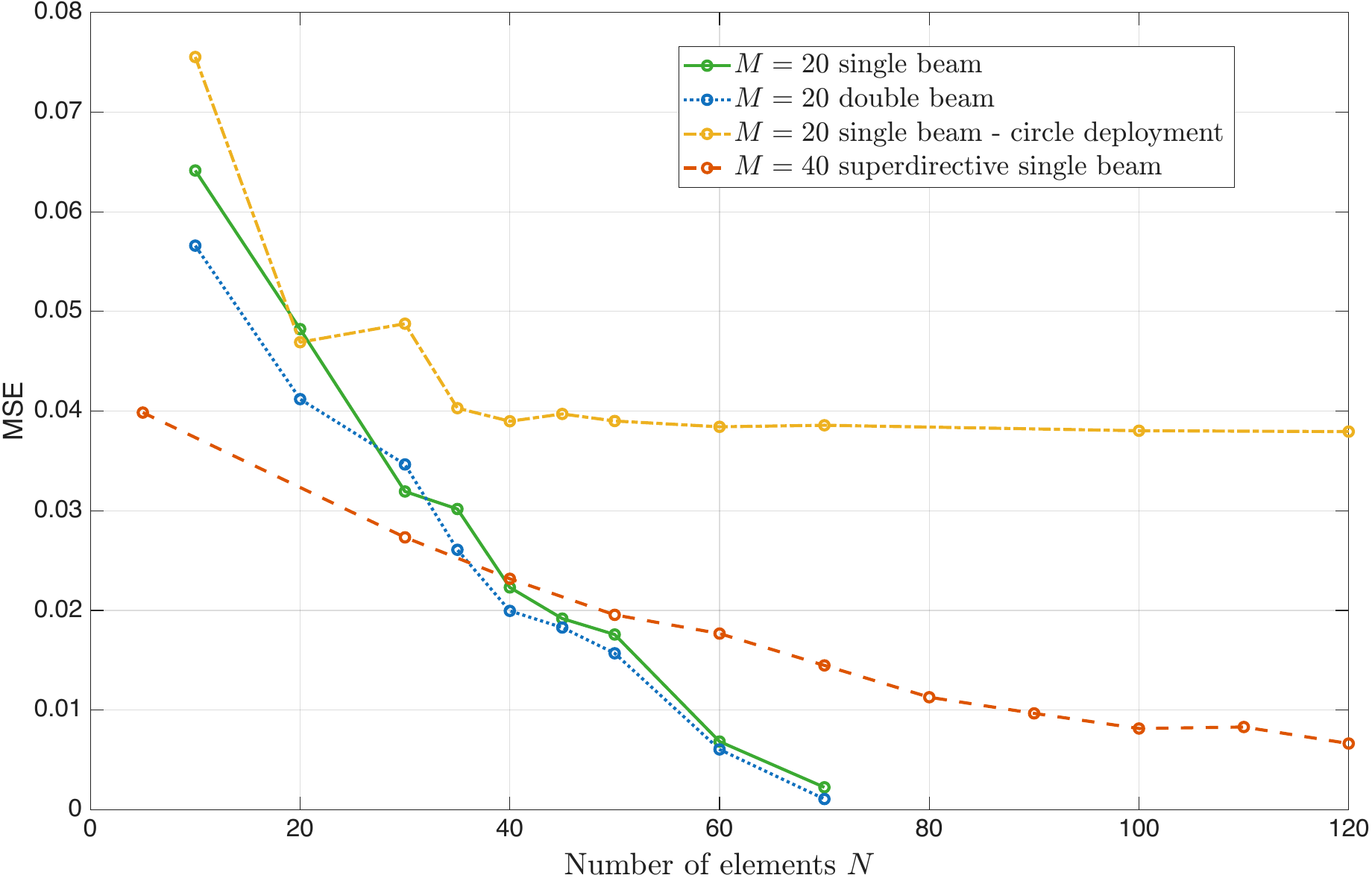}
\caption{Average \ac{MSE} as a function of the number of elements $N$ for the scenarios investigated in Figs. \ref{fig:Pattern}, \ref{fig:Patternbis}, and \ref{fig:PatterSuper} with random deployment, as well as for the elements deployment on a circle.} 
\label{fig:MSE}
\end{figure}

\section*{Acknowledgment}
The author would like to thank Nicol\'o Decarli and Mattia Fabiani for the useful discussions. 

\ifCLASSOPTIONcaptionsoff
\fi
\bibliographystyle{IEEEtran}

\bibliography{IEEEabrv,Biblio/BiblioDD,Biblio/MetaSurfaces,Biblio/EMInformationTheory,Biblio/IntelligentSurfaces,Biblio/MassiveMIMO,Biblio/MIMO,Biblio/THzComm,Biblio/EMTheory,Biblio/WINS-Books,Biblio/Vari}

\end{document}